\documentclass[conference]{IEEEtran}
\IEEEoverridecommandlockouts
\usepackage{cite}
\usepackage{amsmath,amssymb,amsfonts}
\usepackage{algorithmic}
\usepackage{graphicx}
\usepackage{textcomp}
\usepackage{xcolor}
\usepackage{amsthm}
\newtheorem{theorem}{Theorem}
\newtheorem{proposition}{Proposition}

\def\BibTeX{{\rm B\kern-.05em{\sc i\kern-.025em b}\kern-.08em
    T\kern-.1667em\lower.7ex\hbox{E}\kern-.125emX}}
\begin{document}

\title{Optimal UAV Deployment for Rate Maximization in IoT Networks\vspace{-10pt}
}

\author{Maryam Shabanighazikelayeh and Erdem Koyuncu \\ 
Department of Electrical and Computer Engineering, University of Illinois at Chicago\vspace{-15pt}} 



\maketitle

\begin{abstract}
We consider multiple unmanned aerial vehicles (UAVs) at a common altitude serving as data collectors to a network of IoT devices. First, using a probabilistic line of sight channel model, the optimal assignment of IoT devices to the UAVs is determined. Next, for the asymptotic regimes of a large number of UAVs and/or large UAV altitudes, we propose closed-form analytical expressions for the optimal data rate and characterize the corresponding optimal UAV deployments. We also propose a simple iterative algorithm to find the optimal deployments with a small number of UAVs at high altitudes. Globally optimal numerical solutions to the general rate maximization problem are found using particle swarm optimization.
\end{abstract}

\begin{IEEEkeywords}
UAV-aided communications, rate maximization.\vspace{-10pt} 
\end{IEEEkeywords}

\section{Introduction}
Unmanned aerial vehicles (UAVs) have been recently utilized in a variety of applications. For example, UAVs can serve as base stations providing service for mobile users \cite{myoutage, myquant, kalant}. A similar 
use case is UAVs as data collection units \cite{Flight,striking}, especially in the context of Internet of Things (IoT) applications \cite{data,IoTtime,IoTNN}. In \cite{IoTtime}, the authors study UAVs as data collectors from time-constrained IoT devices for
offloading excessive traffic of existing wireless networks. Another example is \cite{Uavvr}, which investigates UAV-IoT data capture and networking for remote scene virtual reality immersion.



Energy efficiency is a fundamental issue in UAV-aided IoT networks as both the UAVs and the IoT devices typically have severe battery and power limitations. Several solutions have thus been proposed to address the energy efficiency challenges of UAV-aided IoT networks \cite{IoTMoz,Rechargable,striking, anotheriot, Erelay}. In particular, \cite{striking} studies the tradeoffs between the energy efficiency of the ground IoT sensors and the overall system throughput by optimizing various system parameters including the UAV flying speeds and altitudes. In \cite{anotheriot}, the authors consider the hovering altitude and power allocation problem for a three tier network consisting of satellites, UAVs, and the IoT devices. The power efficiency provided by multiple UAV relays between a density of IoT devices and base stations is studied in \cite{Erelay}.

Trajectory optimization and optimal deployment
of UAVs is another important problem in designing UAV-aided systems\cite{myquant,IoTtime,3D, Online,CoMP,PIMRC}. In general, this class of problems are non-convex
optimization problems in which dimensionality increases with the number of UAVs. Hence, providing a globally optimal solution is very challenging. Several different optimization methods have  been proposed, including evolutionary algorithms \cite{PSOnew,myoutage}. In \cite{myquant}, the authors propose a quantization theory approach to solve the deployment and trajectory optimization problem. However, the used communication model is a line of sight (LOS) model and does not consider the non line of sight (NLOS) effects \cite{Alhourani}. We refer to \cite{koyuncu1, koyuncu2} for other applications of quantization theory to the deployment of non-UAV networks. In \cite{Rechargable}, the authors consider a cooperative approach to provide coverage and
long term information services for IoT nodes in UAV-aided networks. The authors divide the original non-convex problem into three subproblems and use a block coordinate descent-based iterative
algorithm to solve mentioned subproblems. In \cite{IoTtime}, the authors jointly optimize the UAV trajectory and
the radio resource allocation to serve the maximum number of IoT devices. Globally optimal solutions are found for small scale scenarios using the branch, reduce and
bound algorithm, and suboptimal algorithms are developed for larger scale scenarios.


Most of the previous works rely on a numerical approach to solve the UAV deployment problems in IoT networks. In addition, in some works, the communication model is too simple and does not capture NLOS attenuation. In this work, we consider a probabilistic LOS model and formulate the rate maximization problem accordingly. We find the optimal assignment of the IoT nodes to the data collector UAVs. In addition, for the asymptotic regimes of either a large number of UAVs or large UAV altitudes, we find the optimal deployment of UAVs, and the corresponding optimal data rates. We also verify our analysis with numerical simulations conducted using the particle swarm optimization (PSO) algorithm.

The rest of this paper is organized as follows: In Section \ref{sec1}, we introduce the system model. In Section
\ref{sec2}, we study the optimal assignment of IoT nodes to their UAVs. We also present our asymptotic analysis on the optimal placement of UAVs and corresponding data rates.  In Section  \ref{sec4}, we present the numerical
simulation results. Finally, in Section \ref{sec5}, we draw our
main conclusions and discuss future work. Some of the
technical proofs are provided in the appendices.

\section{System Model and Problem Formulation}
\label{sec1}
Let $q$ be the location of an IoT device in the $d$-dimensional Euclidean space $\mathbb{R}^d$ where $d \in \{1,2\}$. Also, let $x_i$ be the projection of UAV location on $\mathbb{R}^d$, and $h$ denote a common altitude for the UAVs. In this work, we adopt the probabilistic LOS model for the UAVs, as presented in \cite{Alhourani}. According to this model, there can be LOS communication between UAV $i$ at $(x_i,h)$  and the IoT device at $q$ with a certain probability $P_{LOS}$. Otherwise, the IoT-to-UAV link can only support NLOS communication with probability of $P_{NLOS} = 1-P_{LOS}$. The LOS probability $P_{LOS}$ has an explicit dependence on the distances as defined through
\begin{equation}
\label{plos}
P_{LOS}(\|x_i-q\|) \triangleq \frac{1}{1+ce^{-b(\mathrm{tan}^{-1}(\frac{h}{\|x_i-q\|})-c)}},
\end{equation}
where $b$ and $c$ are environment-dependent parameters. An example scenario consisting of one IoT device communicating with two UAVs is illustrated in Fig. \ref{sysmodel}.

Once again following \cite{Alhourani}, let us assume that the NLOS path incurs an extra attenuation of $\delta$ compared to the LOS path, where $0 < \delta < 1$. In such a scenario, using Shannon's well-known capacity formula for the Gaussian channel, the achievable data rate between the  IoT device at $q$ and the UAV at $x_i$ can be epxressed as
\begin{multline}
\label{rate}
\!\!R_i(q) =
\mathrm{log}_2\!\left(\!1\!+\!\frac{\rho A}{N_0 (\|x_i-q\|^2+h^2)^{r/2}}\right) \!P_{LOS}(\|x_i-q\|) \\
 + \mathrm{log}_2\!\left(1\!+\!\frac{\rho A\delta}{N_0(\|x_i-q\|^2+h^2)^{r/2}}\right) \!P_{NLOS}(\|x_i\!-\!q\|),\!
\end{multline}
where $\rho$ is the fixed power of IoT devices,  $r$ is the path loss exponent, $N_0$ is the noise power, and $A$ is a constant which depends on the system parameters such as operation frequency and antenna gain \cite{Azari2018}. Obviously, it is optimal for each IoT device to connect to the UAV that will maximize its data rate. In other words, an IoT device at location $q$ should be connected to the UAV with index
\begin{align}
\label{bestindex} I^{\star}(q) \triangleq \arg\max_i R_i(q)
\end{align}


The maximum data rate that can be provided to the IoT device is then $\max_i R_i(q)$. Suppose now that the IoT devices are distributed over the area of interest according to a certain density function $f(q)$, where $\int_{\mathbb{R}^d} f(q)\mathrm{d}q = 1$. 
Averaging out the maximum data rate of an IoT device $\max_i R_i(q)$ over the IoT device density $f$, the maximum achievable data rate between the IoT devices and the UAVs are given by
\begin{multline}
\label{meanrate}
R(X,f) = \\\int_{\mathbb{R}^{d}}\mathrm{max}_{i}\biggl[\mathrm{log}_2\biggl(1+\frac{\gamma}{(\|x_i-q\|^2+h^2)^{r/2}}\biggr) P_{LOS}(\|x_i-q\|)  \\ \!\!\!+\!
\mathrm{log}_2\biggl(\!1\!+\!\frac{\gamma \delta}{(\|x_i\!-\!q\|^2\!+\!h^2)^{r/2}}\!\!\biggr)\! P_{NLOS}(\|x_i\!-\!q\|)\biggl] \!f(q) \mathrm{d}q,\!\!\!
\end{multline}
where $\gamma \triangleq \frac{A \rho}{N_0}$ and $X = [x_1\,\,x_2 \cdots x_n]$ is the UAV deployment. The goal of this paper is to find the optimal deployment $X$ such that $R(X,f)$ is maximized. In other words, we wish to find  the solution to the following optimization problem:
\begin{equation}
\label{objectivefuncDep}
X^{\star} = [x_1^{\star} \,\, x_2^{\star}\, \cdots\, x_n^{\star}] =  \mathrm{argmax}_{\mathrm{X}} R(X,f)
\end{equation}
In the following, we first determine an explicit expression for the optimal UAV assignment to each IoT device (\ref{bestindex}). We will then focus on the asymptotic regimes of a large number of UAVs or high UAV altitudes to analytically solve the optimal UAV deployment problem as given by (\ref{objectivefuncDep}).


\begin{figure}
\center
\includegraphics[scale=.18]{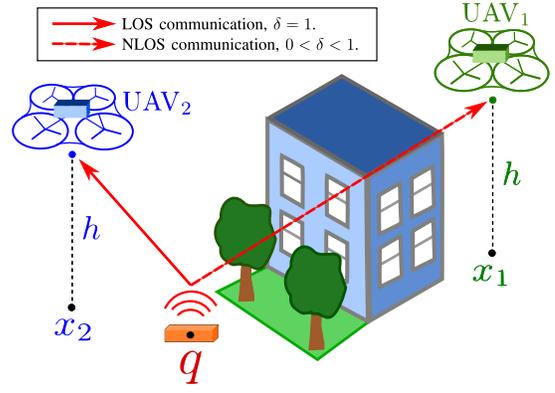}\vspace{-7pt}
\caption{An IoT device communicating with two UAVs over possible LOS and NLOS channels.}
\label{sysmodel}
\end{figure}


\section{Optimal Placement of UAVs}
\label{sec2}
In this section, we present our main analytical results. We first determine the optimal assignment of IoT devices to UAVs. To gain initial insight on this problem, first consider the simple scenario of a pure LOS model, where we consider $P_{LOS} = 1$ and $P_{NLOS} = 0$, independently of the locations of the UAVs and the IoT device. In this case, according to (\ref{rate}) and (\ref{bestindex}), the optimal UAV assignment evaluates to 
\begin{align}
I^{\star}(q) & = \arg\max_i  \mathrm{log}_2\biggl(1+\frac{\gamma}{(\|x_i-q\|^2+h^2)^{r/2}}\biggr) \\
\label{nearestuav} & = \arg\min_i \|q - x_i\| \quad (P_{LOS} = 1,\,P_{NLOS} = 0).
\end{align}
In other words, each IoT device should be connected to its closest UAV. However, in our probabilistic LOS model, the same conclusion cannot be reached immediately, due to the non-trivial dependence of the LOS probabilities and the rate expressions on the IoT-to-UAV distances. Nevertheless, connecting each IoT device to its closest UAV, i.e., the assignment rule in (\ref{nearestuav}) still turns out to be optimal in the case of the probabilistic LOS model, as the following proposition shows.



\begin{proposition}
\label{prop1}
With the probabilistic LOS model, the maximum rate is achieved when each IoT device is connected to the closest UAV. In other words, $I^{\star}(q) = \arg\min_i \|q - x_i\|$.
\end{proposition}
\begin{proof}
Let $d = \|q - x_i\|$. According to (\ref{rate}), we have
\begin{align}
\nonumber R_i(q) & = \mathrm{log}_2\left(1+\frac{\gamma}{(d^2+h^2)^{r/2}}\right) P_{LOS}(d)  + \\ & \label{qpweuqpoweuq1} \qquad \mathrm{log}_2\left(1+\frac{\gamma\delta}{(d^2+h^2)^{r/2}}\right) (1- P_{LOS}(d)) \\ \label{qpweuqpoweuq2} & =  
L_1(d) P_{LOS}(d) + L_2(d),
\end{align}
where
\begin{align}
L_{1}(d) & = \mathrm{log}_2\left(1+\frac{\gamma(1-\delta)}{(d^2+h^2)^{r/2}+\gamma\delta}\right), \\
L_{2}(d) & =  \mathrm{log}_2\left(1+\frac{\gamma\delta}{(d^2+h^2)^{r/2}}\right) 
\end{align}
and $P_{LOS}(d)$ is as defined in (\ref{plos}). The equality of (\ref{qpweuqpoweuq1}) and (\ref{qpweuqpoweuq2}) can be verified through straightforward algebraic manipulations.  The result then follows as $L_1, P_{LOS}, L_2$ are all monotonically decreasing functions of their arguments.
\end{proof}

Now, let $\nu_i = \{q : \|x_i-q\| \leq \|x_j-q\|, \forall j \neq i
\}$ denote the Voronoi region corresponding to UAV $i$. Then, according to Proposition \ref{prop1}, the IoT device $q\in \nu_i$ should be connected to UAV $i$ to maximize the average data rate. 

We can now optimize the UAV deployment and determine the corresponding best possible average IoT data rates. Our main result in this context is  the following theorem.

\begin{theorem}
\label{thm2}
For asymptotically large UAV altitudes $h$ and/or a large number of UAVs, the optimal deployment of UAVs is derived by solving the following optimization problem:
\begin{align}
\label{objectiveNew}
X^{\star} & \textstyle  = \mathrm{argmin}_{\mathrm{X}} \int \mathrm{min}_{i}\|x_i-q\| f(q) \mathrm{d}q \\  & \textstyle =\mathrm{argmin}_{\mathrm{X}}  \sum_{i=1}^{n} \int_{\nu_i} \|x_i-q\| f(q) \mathrm{d}q
\end{align}
The corresponding optimal data rate is
\begin{multline}
\label{optLargeN}
R(X^{\star},f) = \mathrm{log}_2\left(1+\frac{\gamma}{h^r}\right)\frac{1}{1+c'}+\mathrm{log}_2\left(1+\frac{\gamma \delta}{h^r}\right)\frac{c'}{1+c'} \\
- \frac{bc'}{h(1+c')^2}\mathrm{log}_2\left(\frac{\gamma \delta+h^r}{\gamma + h^r}\right) \sum_{i=1}^{n} \int_{\nu_i}\|x_i^{\star}-q\| f(q)  \mathrm{d}q \\
+ \mathrm{log}_2\left(\frac{\gamma \delta+h^r}{\gamma + h^r}\right) \sum_{i=1}^{n} \int_{\nu_i}o\left(\frac{\|x_i^{\star}-q\|}{h}\right) f(q)  \mathrm{d}q,
\end{multline}
where $c' \triangleq c e^{-b(\frac{\pi}{2}-c)}$.
\end{theorem}
\begin{proof}
See Appendix \ref{appthm2}.
\end{proof}

An interesting byproduct of Theorem \ref{thm2} is that for large number of UAVs and/or arbitrary number of UAVs at high altitudes, the optimal placement is derived from (\ref{objectiveNew}) which is independent of $h$ and $\delta$. Hence, the optimal placement in the mentioned asymptotic regimes is not a function of altitude or attenuation. In addition, the problem of finding the optimal deployment is reduced to solving (\ref{objectiveNew}), for which many methods and results are already available, especially from the quantization theory literature. Once a solution to (\ref{objectiveNew}) is obtained, it can be substituted to (\ref{optLargeN}) to obtain an asymptotically tight expression for the data rates. We now discuss two methods to solve (\ref{objectiveNew}). The first theoretical method provides an analytical solution for the asymptotic regime of a large number of UAVs. The second numerical method will be applicable to any number of UAVs.


\subsection{Quantization Theory Approach}
We first present an analytical approach to solve (\ref{objectiveNew}). We note that  (\ref{objectiveNew}) can be interpreted as the average $\ell_1$-norm distortion of a quantizer with reproduction points $x_1,\ldots,x_n$ for a given source density $f$ \cite{myquant}. As $n\rightarrow\infty$, the optimal UAV deployment in (\ref{objectiveNew}) can be characterized in terms of a density function of UAVs, rather than the individual locations of each UAV. To that end, consider a point density function $\lambda(q)$ such that the cube $[q, q + \mathrm{d}q]$ of volume $dq$ contains $n\lambda(q) dq$ reproduction points (UAVs) with $\int_{\mathbb{R}^d} \lambda(q) \mathrm{d}q = 1$. According to the classical results of quantization theory \cite{zador, glquant}, the optimal point (UAV) density function is as follows:
\begin{equation}
\label{optimalX}
\mathrm{\lambda^{\star}}(q,f) = \textstyle f^{\frac{d}{d+1}}(q) / \int_{R^d} f^{\frac{d}{d+1}}(q') dq'
\end{equation}
Hence, as ${n\to\infty}$, for any $q$, the infinitesimal $[q, q + dq]$ should contain $n \mathrm{\lambda^{\star}}(q,f) dq$ UAVs in an optimal deployment. Furthermore, also using the results in \cite{zador, glquant}, the corresponding optimal value of (\ref{objectiveNew}) can be derived in closed-form as 
\begin{equation}
\label{opt1}
\mathrm{min}_{\mathrm{X}} \! \textstyle\int \mathrm{min}_{i}\|x^{\star}_i-q\| f(q) dq \!=\! k_d n^{-\frac{1}{d}} \|f \|_{\frac{d}{d+1}} \!+\! o(n^{-\frac{1}{d}}),
\end{equation}
where $\|f \|_{\alpha} \triangleq (\int_{\mathbb{R}^d} (f(q))^{\alpha} dq)^{\frac{1}{\alpha}}$ is the $\alpha$-norm of the density f and $k_1$ and $k_2$ are the normalized first moments of the origin-centered interval and the origin-centered regular hexagon, respectively. The normalized $\ell$th moment of an arbitrary origin-centered $A\subset \mathbb{R}^d$ is defined as
\begin{equation}
 m(A) \triangleq \textstyle \int_{A} \|q\|^{\ell} dq /  (\int_{A} dq)^\frac{d+\ell}{d}.
\end{equation}
In particular, for the interval and the regular hexagon, which correspond to the optimal Voronoi cell shapes in one and two dimensions respectively, the normalized first moments can be calculated to be $k_1 = \frac{1}{4}$ and $k_2 = \frac{4+\log 27}{12^{\frac{3}{4}}3}$, respectively. 


Equation (\ref{opt1}) provides a complete asymptotic characterization of the achievable date rate for Theorem \ref{thm2}, because the closed forms of (\ref{optLargeN}) are immediately calculated by substituting the optimal value of $\mathrm{min}_{\mathrm{X}} \int \mathrm{min}_{i}\|x^{\star}_i-q\| f(q) dq$ from (\ref{opt1}) to (\ref{optLargeN}). The final result is summarized via the following theorem.

\begin{theorem}
\label{thm3}
For an asymptotically large number of UAVs, the optimal UAV point density function that maximizes the data rate is given by (\ref{optimalX}). The corresponding optimal data rate is 
\begin{multline}
\label{optLargefinal}
R(X^{\star},f) = \mathrm{log}_2\left(1+\frac{\gamma}{h^r}\right)\frac{1}{1+c'}+\mathrm{log}_2\left(1+\frac{\gamma \delta}{h^r}\right)\frac{c'}{1+c'} \\
\!\!- \frac{bc'}{h(1\!+\!c')^2}\mathrm{log}_2\!\left(\frac{\gamma \delta\!+\!h^r}{\gamma\! +\! h^r}\right) \!\!\left[k_d n^{-\frac{1}{d}} \|f \|_{\frac{d}{d+1}}\!\! +\! o(n^{-\frac{1}{d}})\right]\!.\!\!
\end{multline}
\end{theorem}

This provides a complete asymptotic characterization of the rate for large number of UAVs. Unfortunately, the knowledge of the optimal density function of the UAVs does not immediately lead to the knowledge of the optimal discrete UAV locations. However, for the special case of one dimension, the optimal discrete placement of UAVs can also be approximated using a variant of inverse transform sampling \cite{myquant}: Let $X^{\star} = [x_1^{\star}\, x_2^{\star} \cdots x_n^{\star}]$ be the optimal deployment. Suppose $x_1^{\star} \leq x_1^{\star} \leq ... \leq x_n^{\star}$ without loss of generality. For $x\!\in\![0,1]$, let $\Lambda^{\star}_{inv}(x,f)$ be the unique real number that satisfies
\begin{equation}
\label{lambdaeq}
\textstyle \int_{0}^{\Lambda^{\star}_{inv}(x,f)} \lambda^{\star}(q,f) dq = x.
\end{equation}
 Then, $x_i^{\star}$ can be approximated as 
\begin{equation}
\label{xopt}
 \textstyle x_i^{\star} \simeq \Lambda^{\star}_{inv}\left(\frac{2i-1}{2n},f\right)
\end{equation}
Hence, to find the optimal placement of UAVs, we can first solve (\ref{lambdaeq}) for $\Lambda^{\star}_{inv}(x,f)$ and then use (\ref{xopt}) to calculate the optimal UAV locations. For two dimensions, or a non-asymptotic number of UAVs, we consider a numerical solution to (\ref{objectiveNew}). Details of the solution are described in what follows.


\subsection{Iterative Approach}
In this numerical approach to solving (\ref{objectiveNew}), the UAV locations $x_{1,0},...,x_{n,0}$ are first initialized randomly at Iteration $0$. We then perform the following procedure iteratively, essentially considering a generalized Lloyd algorithm \cite{lbgalgo} for the $\ell_1$-norm distortion measure. At Iteration $k$, where $k \geq 1$, we first calculate the Voronoi regions
\begin{equation}
\nu_{i,k} = \{q : \|x_{i,k-1}-q\| \leq \|x_{j,k-1}-q\|, \forall j \neq i
\},
\end{equation}
 Keeping the Voronoi regions fixed, we then solve the following optimization problem to update the optimal solution $X$: 
\begin{equation}
\label{iterativexopt}
X_{k} = \mathrm{argmin}_X \sum_{i=1}^{n} \int_{\nu_{i,k}} \mathrm \|x_{i}-q\| f(q)  \mathrm{d}q.
\end{equation}
Solving (\ref{iterativexopt}) is equivalent to solving the optimization problem 
\begin{equation}
\label{1234}
x_{i,k} = \mathrm{argmin}_{x_i} \int_{\nu_{i,k}} \mathrm \|x_{i}-q\| f(q)  \mathrm{d}q
\end{equation}
for each $i\in\{1,2,..., n\}$. The problem (\ref{1234}) is a convex optimization problem, as the objective function is the positive weighted summation of convex norms. Therefore, we can solve (\ref{1234}) by using any globally optimal approach such as gradient descent. Furthermore, for one dimension, we can provide a closed form solution for (\ref{1234}) by solving
\begin{equation}
\frac{\partial}{\partial x_i} \int_{\nu_{i,k}} \!\! \|x_{i}-q\| f(q)   \mathrm{d}q =
\int_{\nu_{i,k}}\!\! \mathrm{sign}(x_i-q) f(q) \mathrm{d}q = 0
\end{equation}
Solving for $x_i$, we obtain $x_{i,k} = \mathrm{median}(f_c(q))$, where $f_c(q) \triangleq \frac{f(q)}{\int_{\nu_{i,k}} f(q) dq},\; q \in \nu_{i,k}$.
 
Note that one can also attempt to directly solve (\ref{meanrate}) in an iterative fashion. The calculation of the Voronoi regions $\nu_{i,k}$ remains the same as it is optimal for each IoT device to be connected to its closest UAV. We can update the UAV locations as $x_{i,k} = \arg\max_{x_i} \int_{\nu_{i,k}} R_i(q) \mathrm{d}q$. The end result is an iterative ascent algorithm for the original objective function, which is very much desirable. On the other hand, the problem with this approach is that the optimization of $x_{i,k}$ still remains non-convex. The strength of our iterative approach stems from the fact that it convexifies the entire optimization, resulting in a very fast implementation. The numerical simulations in the next section also show that our  convexification approach results in only negligible loss of performance.

\section{Numerical Results}
\label{sec4}
In this section, we provide numerical simulation results that confirm our analytical findings. For a general approach that is applicable to all scenarios, we used the PSO
method \cite{PSO} to solve the optimization problem (\ref{objectivefuncDep}). 

The PSO method is a population-based iterative algorithm
for solving non-convex optimization problems. In general, population-based optimization algorithms such as PSO are known to outperform the simpler gradient
descent like approaches. Specifically, multiple candidate solutions (population agents)  helps to avoid
locally optimal solutions. This makes PSO-like algorithms particularly suitable for  multiple-UAV optimization problems \cite{myoutage} which are complicated non-convex problems in general. 

We provide simulation results to validate Theorems \ref{thm2} and \ref{thm3} by deriving the optimal solution of (\ref{objectiveNew}) using quantization theoretical and iterative approaches. We also investigate the effects of altitude and attenuation factors on the achievable rates. For our numerical simulations, we have used $b = 0.43$, $c = 4.88$, $\gamma = 50dB$, $r = 2$, unless specified otherwise. Also, in the figures, ``Quantization Theory approach'' refers to the results of Theorem \ref{thm3}, while ``Iterative approach'' refers to Theorem \ref{thm2} where the optimal deployment is calculated via the iterative algorithm  in Section III.B.



\begin{figure}[h]
\center
\includegraphics[scale=.55]{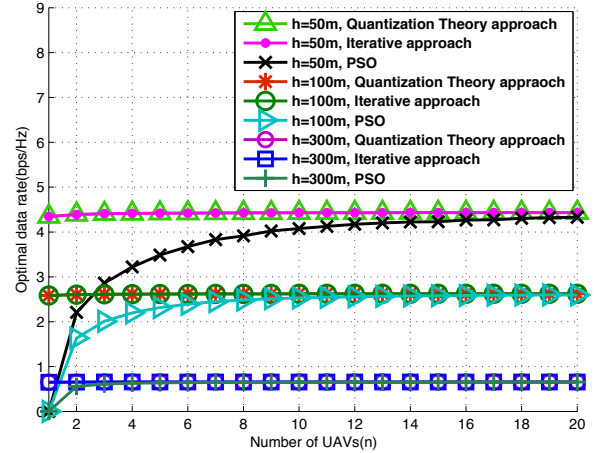}\vspace{-7pt}
\caption{Comparison of UAV deployment algorithms for a one-dimensional uniform density at different altitudes and $\delta = 0.5$.}
\label{1dimuni}
\end{figure}

\begin{figure}[h]
\center
\includegraphics[scale=.55]{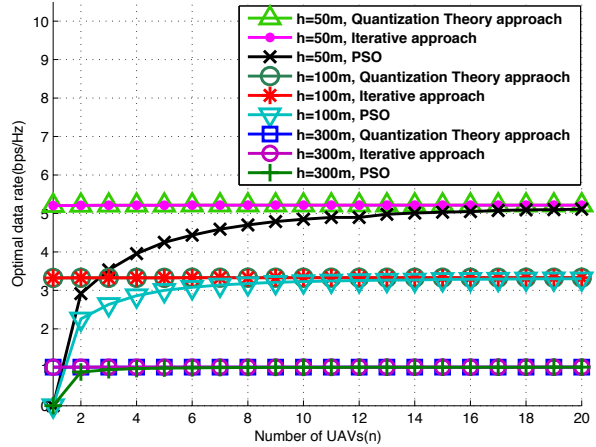}\vspace{-7pt}
\caption{Comparison of UAV deployment algorithms for a one-dimensional uniform density at different altitudes and $\delta = 0.9$.}
\label{1dimuni9}
\end{figure}

Fig. \ref{1dimuni} shows the optimal rate derived with Theorem \ref{thm2} in comparison with results provided by PSO method for different value of altitudes $h$ and different number of UAVs. The horizontal axis represents the number of UAVs, and the vertical axis represents the  data rate. One dimensional uniform density $f(q) = 10^{-3},\; q \in [0,1000]m$ is considered for IoT density. We can observe that for $n>3$ and $h=300m$ which can be considered as a relatively high altitude, the results of Theorem \ref{thm2}, which are applicable to high altitudes matches the exact results derived by solving the original optimization problem (\ref{objectivefuncDep}) using PSO. Furthermore, for a large number of UAVs and any altitude, Theorem \ref{thm3} provides almost the same results as the exact solution of (\ref{objectivefuncDep}). The mentioned scenarios confirm the accuracy of Theorem \ref{thm3}. 


A key observation from  Fig. \ref{1dimuni} is that the optimal data rate converges as the number of UAVs increases. This is more obvious for the case with $h= 300m$. Accordingly, we can conclude that adding more UAVs will not improve the system performance noticeably after some point which depends on the altitude. Specifically, as the altitude increases, the  optimal results are achievable with less number of UAVs.




In Fig. \ref{1dimuni9}, we consider the setup of Fig. \ref{1dimuni} with attenuation factor ($\delta = 0.9$). Similar observations and conclusions as the previous figure can be made. This shows the flexibility of our framework for different environment with variable attenuation.

\begin{figure}[h]
\center
\includegraphics[scale=.42]{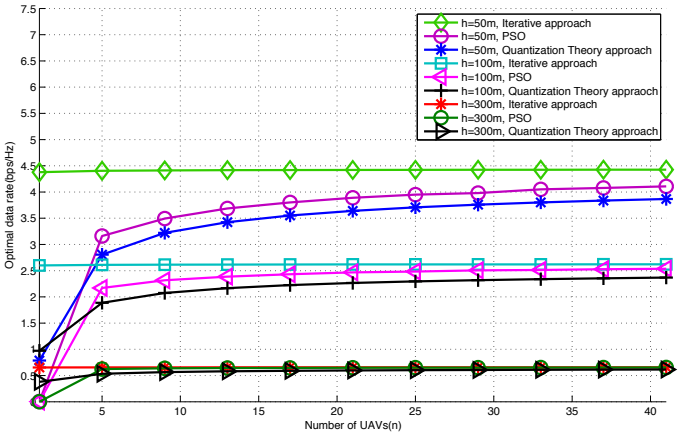}\vspace{-7pt}
\caption{Comparison of UAV deployment algorithms for a two-dimensional Gaussian density at different altitudes and $\delta = 0.5$.}
\label{twodimexample}
\end{figure}

In Fig. \ref{twodimexample}, we consider a two-dimensional Gaussian density with zero mean and covariance matrix $100\cdot\mathbf{I}$, where $\mathbf{I}$ is the identity matrix. Similar conclusions can be made as compared with the one dimensional examples: At high altitudes both the quantization theoretical and the iterative approaches provides a close approximation to the exact performance as provided the the PSO algorithm. At low altitudes, as the number of UAVs grow to infinity, the approximations again converge to the optimal performance. An interesting difference is that the quantization theoretical approach provides a better approximation than the iterative approach when the number of UAVs are small. A more precise theoretical analysis is needed to understand this phenomenon.

Consider now a time-varying IoT device density
$f_t(q) = (1+2|t|)(q - 2 + 2|t|)^{2|t|},\; q \in [2 - 2|t|, 3 -2|t|]$, with $5$ UAVs, where $t \in [-1, 1]$ represents the time index. At each time, we can optimize the UAV deployment to come up with the optimal UAV trajectories for the time interval $[-1,1]$. According to (\ref{xopt}), the optimal trajectory of UAV $i$ can be approximated as
\begin{equation}
\label{xoptnew}
 x_{i,t}^{\star} \simeq \Lambda^{\star}_{inv}\left(\frac{2i-1}{2n},f_t\right).
 \end{equation}
In order to calculate the optimal UAV trajectories, we need to first derive $\Lambda^{\star}_{inv}$ from (\ref{lambdaeq}). Using (\ref{optimalX}), we first obtain 
\begin{equation}
\lambda_{t}^{\star}(q,f) = (1+|t|)(q - 2 + 2|t|)^{|t|}
\end{equation}
Hence, $\Lambda^{\star}_{inv}$ can be calculated as
\begin{equation}
\Lambda^{\star}_{inv}(x,f_t) = 2-2|t| + x^{\frac{1}{1+|t|}}    
\end{equation}
Accordingly, the optimal trajectory of UAV $i$ can be approximated by
\begin{equation}
 \label{qttraj} x_{i,t}^{\star} \simeq  2-2|t| + \left(\frac{2i-1}{2n}\right)^{\frac{1}{1+|t|}}  
\end{equation}
 
 Fig. \ref{figtraject} illustrates the optimal trajectories provided by the PSO method, quantization theory (\ref{qttraj}),   and the iterative approach. The trajectories provided by the iterative approach and the quantization theory approach of (\ref{qttraj}) are almost the same. Both trajectories are slightly different than the trajectory provided by the PSO algorithm. These results show that for the asymptotic scenarios (high altitudes or large number of UAVs) where Theorems \ref{thm2} and \ref{thm3} become valid, we may use either the quantization theoretical or the iterative approach to calculate the optimal UAV deployments without great loss in performance. This way, we avoid running the computationally expensive PSO algorithm (or a similar globally optimal optimization algorithm) to solve the original problem in (\ref{objectivefuncDep}).
 
  \begin{figure}[h]
\center
\includegraphics[scale=.53]{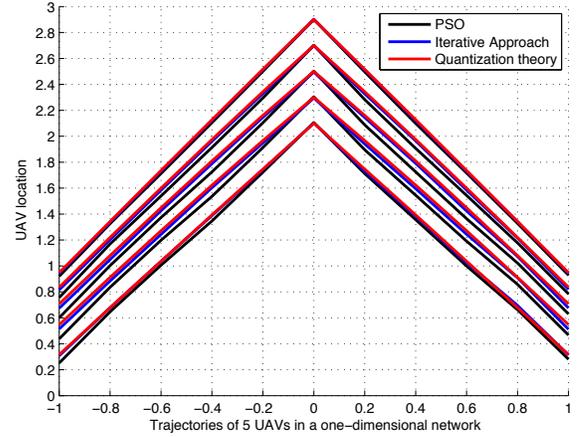}\vspace{-7pt}
\caption{Trajectories of 5 UAVs in a one-dimensional network.}
\label{figtraject}
\end{figure}


\section{Conclusion}
\label{sec5}
We have studied the optimal deployment of UAVs
serving as data collectors from time constrained IoT devices. Our
objective has been to maximize the collected data in an specified time by maximizing the communication data rate. We provided the optimal solution of IoT device-UAV association problem. Furthermore, we approximated the original non-convex problem with multiple convex problems and provided quantization theory based closed form solutions. We also proposed an iterative approach to solve the approximated problem. Finally, we compared the results of the proposed approaches with the results derived by solving the original non-convex problem. The simulation results shows the flexibility of proposed approaches for different practical scenarios.

\section*{Acknowledgement}
This work was supported in part by the NSF Award
CCF–1814717.

\appendices

\section{Proof of Theorem \ref{thm2}}
\label{appthm2}
We consider the following asymptotic expansions for different parts of the proof (the expansions are valid for $t \to 0$):
\begin{equation}
\label{555}
\mathrm{log}_2(a+t) = \mathrm{log}_2(a) + \frac{t}{\mathrm{log}(2)a }+ o\left(\frac{t}{\mathrm{log}(2)a }\right)
\end{equation}
\begin{equation}
\label{powerr}
(1+t)^{r} = 1 + rt + o(t)
\end{equation}
\begin{equation}
\label{tanlim}
 \mathrm{tan}^{-1}(1/t) = \pi/2 - t + o(t)
\end{equation}
\begin{equation}
\label{explim}
\frac{1}{1+c e^{bt}} = \frac{1}{c+1} - \frac{bct}{(c+1)^2} + o(t)
\end{equation}
We now proceed with the proof of the theorem. Let $d \triangleq \mathrm{min}_{i}\|x_i-q\|$ and $t \triangleq \frac{d}{h}$. In an optimal deployment, for large number of UAVs we have $d \simeq 0$. Therefore, for large number of UAVs and/or high altitudes, $t \simeq 0$ is a valid assumption. Having this assumption, the following is concluded from (\ref{powerr}):
\begin{equation}
\label{2323}
\frac{\gamma}{(d^2+h^2)^{\frac{r}{2}}} = \frac{\gamma}{h^r}\left(1-\frac{r d^2}{2h^2} + o\left(\frac{d^2}{h^2}\right)\right) \end{equation}

Using (\ref{2323}) and  (\ref{555}), we obtain
\begin{multline}
\label{444}
\mathrm{log}_2\left(1+\frac{\gamma}{(d^2+h^2)^{r/2}}\right) = \\ \mathrm{log}_2\left(1+\frac{\gamma}{h^{r}}\right) -\frac{r d^2}{2\mathrm{log}2 h^2(\gamma+h^r)} \!+\! o\left(\frac{r d^2}{h^2(\gamma+h^r)}\right) \end{multline}

Furthermore, according to (\ref{tanlim}), we have
\begin{equation}
P_{LOS}(d) = \frac{1}{1+c e ^{-b(\frac{\pi}{2} - \frac{d}{h} + o(\frac{d}{h}) -c)}},
\end{equation}
and, by (\ref{explim}), we obtain
\begin{equation}
\label{plosapprox}
P_{LOS}(d) = \frac{1}{1+c'} - \frac{b c' d}{h(1+c')^2} + o\left(\frac{d}{h}\right),
\end{equation}
where $c' = c e^{-b(\frac{\pi}{2}-c)}$.

Substituting (\ref{444}) and (\ref{plosapprox}) to (\ref{rate}), we have
\begin{multline}
\max_{i} R_{i}(q) = 
\log\left(1+\frac{\gamma}{h^{r}}\right) \frac{1}{1+c'} + \log_2\left(1+\frac{\gamma \delta}{h^{r}}\right) \frac{c'}{1+c'} \\ - \frac{bc'}{h(1+c')^2}\mathrm{log}\left(\frac{\gamma \delta+h^r}{\gamma + h^r}\right) d \!+\! \mathrm{log}\left(\frac{\gamma \delta\!+\!h^r}{\gamma \!+\! h^r}\right) o\left(\frac{d}{h}\right)\!\!\!
\end{multline}
Substituting the value of $d$ and averaging out the IoT device density, we obtain the theorem statement. 


\vspace{12pt}

\end{document}